\documentclass[12pt,a4paper]{article}
\usepackage{amssymb}
\usepackage{amsfonts}
\usepackage[onehalfspacing]{setspace}

\newtheorem{proposition}{Proposition}

\newenvironment{proof}[1][Proof]{\noindent\textbf{#1.} }{\ \rule{0.5em}{0.5em}}
\input{tcilatex}
\begin{document}

\title{On the Behavioral Consequences of Reverse Causality\thanks{%
This paper cannibalizes and supersedes a previous paper titled
\textquotedblleft A Simple Model of Monetary Policy under Phillips-Curve
Causal Disagreements\textquotedblright . Financial support by ERC Advanced
Investigator grant no. 692995 is gratefully acknowledged. I also thank Yair
Antler, Tuval Danenberg, Nathan Hancart and Heidi Thysen for helpful
comments.}}
\author{Ran Spiegler\thanks{%
Tel Aviv University, University College London and CFM. URL:
http://www.tau.ac.il/\symbol{126}rani. E-mail: rani@tauex.tau.ac.il.}}
\maketitle

\begin{abstract}
Reverse causality is a common causal misperception that distorts the
evaluation of private actions and public policies. This paper explores the
implications of this error when a decision maker acts on it and therefore
affects the very statistical regularities from which he draws faulty
inferences. Using a quadratic-normal parameterization and applying the
Bayesian-network approach of Spiegler (2016), I demonstrate the subtle
equilibrium effects of a certain class of reverse-causality errors, with
illustrations in diverse areas: development psychology, social policy,
monetary economics and IO. In particular, the decision context may protect
the decision maker from his own reverse-causality causal error. That is, the
cost of reverse-causality errors can be lower for everyday decision makers
than for an outside observer who evaluates their choices.\bigskip \bigskip
\pagebreak
\end{abstract}

\section{Introduction}

Reverse causality is a ubiquitous error. Observing a correlation between two
variables of interest, we often form an instinctive causal interpretation in
one direction, yet true causation may go in the opposite direction. Reverse
causality is often invoked as a warning to social scientists and other
researchers, to beware of naive causal interpretations of correlational
data. For example, Harris (2011) famously criticized the developmental
psychology literature for taking for granted that observed correlation
between child personality and parental behavior reflected a causal link from
the latter to the former. Harris argued that causality might go in the
opposite direction: parental behavior could be a \textit{response} to the
child's innate temperament.

Clearly, when researchers and other outside observers interpret correlations
under the sway of a reverse-causality error, their evaluation of private
interventions or public policy may become distorted. But what happens when
decision makers $act$ on a reverse-causality error, such that their
resulting behavior affects the very statistical regularities from which they
draw inferences through their faulty prism?

For illustration, consider the developmental psychology example mentioned
above, and embed it in the following decision context. A counselor observes
a child's initial condition and chooses a therapy. These two factors cause
the child's behavior. Parental behavior is a response to the child's
behavior and the counselor's therapy. However, the counselor operates under
the perception that parental behavior is an independent variable that joins
the list of factors that cause the child's behavior. This difference can be
represented by two directed acyclic graphs (DAGs), which conventionally
represent causal models (Pearl (2009)):%
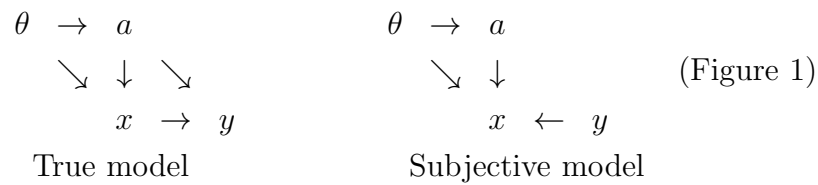
\begin{eqnarray}
&&%
\begin{array}{ccccc}
\theta  & \rightarrow  & a &  &  \\ 
& \searrow  & \downarrow  & \searrow  &  \\ 
&  & x & \rightarrow  & y%
\end{array}%
\qquad \qquad 
\begin{array}{ccccc}
\theta  & \rightarrow  & a &  &  \\ 
& \searrow  & \downarrow  &  &  \\ 
&  & x & \leftarrow  & y%
\end{array}
\TCItag{Figure 1} \\
&&\quad \text{True model}\qquad \qquad \qquad \quad \text{Subjective model} 
\nonumber
\end{eqnarray}%
In this diagram, $\theta $ represents initial conditions (e.g. family
characteristics), $a$ represents the counselor's action, $x$ represents the
child's behavior and $y$ represents the parent's behavior. The counselor's
subjective causal model departs from the true model by inverting the link
between $x$ and $y$ and severing the link from $a$ into $y$. That is, the
counselor regards $y$ as an exogenous variable that causes $x$, whereas in
fact it is an endogenous variable and $x$ is one of its direct cause.

The diagram is paradigmatic, in the sense that it could fit many other
real-life situations. In a monetary economics context, different
specifications of the Phillips curve disagree over which of the two
variables, inflation and employment, is a dependent variable and which is an
explanatory variable. This can be viewed as a disagreement over the
direction of causality (see Sargent (1999) and Cho and Kasa (2015) -- I will
discuss this example in detail in Sections 2 and 3). In a social policy
context, $\theta $ represents initial socioeconomic or demographic
conditions, $a$ represents welfare policy, $x$ represents poverty levels or
income inequality, and $y$ is a public-health indicator. According to this
interpretation, public health is an objective consequence of income
inequality or poverty, yet designers of social policy operate under the
assumption that health is an exogenous, independent factor that affects the
social outcome rather than being caused by it (for a survey of a literature
that wrestles with the causal relation between income inequality and health,
see Pickett and Wilkinson (2015)).

I should emphasize that in neither of these examples do I take an empirical
stand on whether the supposedly \textquotedblleft true
model\textquotedblright\ is indeed true. My only objective is to take
familiar disagreements over the direction of causality in various contexts
and place them in a $decision$ context, in which the objective
data-generating process is affected by the decision maker's behavior, which
in turn is a response to statistical inferences that reflect his causal
error. In this way, reverse causality ceases to be the exclusive problem of
outside observers who make scientific claims and becomes a payoff-relevant
concern for decision makers with skin in the game. Is there a fundamental
difference between the two scenarios? How does the decision context affect
the magnitude of errors induced by reverse causality?

To study these questions, I apply the modeling approach developed in
Spiegler (2016,2020a), which borrows the formalism of Bayesian networks from
the Statistics and AI literature (Cowell et al. (1999), Pearl (2009)) to
analyze decision making under causal misperceptions. According to this
approach, a decision maker (DM henceforth) fits a DAG-represented causal
model to an objective joint distribution over the variables in his model.
The DM best-replies to the possibly distorted belief that arises from this
act of fitting a misspecified causal model to objective data. I apply this
model to the specification of true and subjective causal models depicted in
Figure 1.

I employ a quadratic-normal parameterization, such that the DM's payoffs are
given by a quadratic loss function of $a$ and $x$, and the joint
distribution over the four variables is Gaussian. This parameterization
helps in terms of analytic tractability and transparent interpretation of
the results, which allows for crisp comparative statics. However, it is also
appropriate because fitting a DAG-represented model to a Gaussian
distribution is equivalent to OLS estimation of a recursive system of linear
regressions (Koller and Friedman (2009, Ch. 7)). Causal interpretations of
linear regressions permeate discussions of reverse causality in
social-science settings. Assuming that the environment is Gaussian ensures
that the linearity of the DM's model is not wrong per se -- the
misspecification lies entirely in the underlying causal structure, which is
represented by the DAG.

The analysis of this simple model demonstrates the subtle equilibrium
effects of decision making under reverse causality errors. One effect is
that features of the model that would be irrelevant in a quadratic-normal
model with rational expectations -- specifically, the variances of the noise
terms in the equations for $x$ and $y$ -- play a key role when the DM
exhibits the reverse-causality error given by Figure 1.

The results also illuminate the difference between committing a
reverse-causality error as a \textquotedblleft spectator\textquotedblright\
or as an \textquotedblleft actor\textquotedblright . A special case of the
specification given by Figure 1 is that the action's objective effect on the
outcome variables $x$ and $y$ is null. In this case, the DM can be viewed as
a \textquotedblleft spectator\textquotedblright . This DM's
reverse-causality error induces a wrong prediction of $x$ and he suffers a
welfare loss as a result. Compare this with a DM who is an \textquotedblleft
actor\textquotedblright , in the sense that his action has a non-null direct
effect on $y$ (but still no effect on $x$). I show that when this effect $%
fully$ offsets the direct effect of $x$ on $y$ (in the sense that $E(y\mid
a,x)=x-a$), the DM's equilibrium strategy is as if he has rational
expectations. This DM suffers no welfare loss as a result of his
reverse-causality error. Thus, the decision context protects an
\textquotedblleft actor\textquotedblright\ from his error in a way that it
does not for a \textquotedblleft spectator\textquotedblright . The lesson is
that in some situations, reverse causality may be less of a problem for DMs
than for the scientists who analyze their behavior.

\section{The Model}

A DM observes an exogenous state of Nature $\theta \sim N(0,\sigma _{\theta
}^{2})$ before taking a real-valued action $a$. The DM is an
expected-utility maximizer with vNM utility function $u(a,x)=-(x-a)^{2}$,
where $x$ is an outcome variable that is determined by $\theta $ and $a$
according to the equation%
\begin{equation}
x=\theta -\gamma a+\varepsilon  \label{eq x}
\end{equation}%
Another outcome $y$ is then determined by $x$ and $a$ according to the
equation%
\begin{equation}
y=x-\lambda a+\eta  \label{eq y}
\end{equation}%
The parameters $\gamma ,\lambda $ are constants that capture the direct
effects of $a$ on the outcome variables $x$ and $y$. Assume $\gamma ,\lambda
\in \lbrack 0,1]$, such that these direct effects are of an
\textquotedblleft \textit{offsetting}\textquotedblright\ nature. The terms $%
\varepsilon \sim N(0,\sigma _{\varepsilon }^{2})$ and $\eta \sim N(0,\sigma
_{\eta }^{2})$ are independently distributed noise variables. The ratio of
the variances of these noise terms, denoted%
\begin{equation}
\tau =\frac{\sigma _{\varepsilon }^{2}}{\sigma _{\eta }^{2}}  \label{tau}
\end{equation}%
will play an important role in the sequel.

A (potentially mixed) strategy for the DM is a function that assigns a
distribution over $a$ to every realization of $\theta $. Once we fix a
strategy, we have a well-defined joint probability measure $p$ over all four
variables $\theta ,a,x,y$, which can be written as a factorization of
marginal and conditional distributions as follows:\footnote{%
The factorization is invalid when some terms involve conditioning on zero
probability events. Since this is not going to be a problem for us, I ignore
the imprecision.}%
\begin{equation}
p(\theta ,a,x,y)=p(\theta )p(a\mid \theta )p(x\mid \theta ,a)p(y\mid a,x)
\label{BNfactorization}
\end{equation}%
The term $p(a\mid \theta )$ represents the DM's strategy. This factorization
reflects the causal structure that underlies $p$ and that can be described
by the leftward DAG in Figure 1, which I denote $G^{\ast }$.\footnote{%
Schumacher and Thysen (2020) study a principal-agent model in which the true
causal model has the same form as $G^{\ast }$, although the nodes have
different interpretations (in particular, the action variable corresponds to
the initial node), and the agent's subjective causal model is quite
different.}

I interpret $p$ as a long-run distribution that results from many
repetitions of the same decision problem. Our DM faces a one-shot decision
problem, against the background of the long-run experience accumulated by
many previous generations of DMs who faced the same one-shot problem (each
time with a different, independent draw of the random variables $\theta
,\varepsilon ,\eta $).\medskip

\noindent \textit{The correct-specification benchmark}

\noindent Suppose the DM correctly perceives the true model -- i.e., he has
\textquotedblleft rational expectations\textquotedblright . Then, he will
choose $a$ to minimize%
\[
E[(x-a)^{2}\mid \theta ,a] 
\]%
where the expectation $E$ is taken w.r.t the objective conditional
distribution $p_{G^{\ast }}(x\mid \theta ,a)\equiv p(x\mid \theta ,a)$. This
means that we can plug (\ref{eq x}) into the objective function and use the
fact that $\varepsilon $ is independently distributed with $E(\varepsilon
)=0 $. Hence, the rational-expectations prediction is that the DM will choose%
\begin{equation}
a_{G^{\ast }}(\theta )=\frac{\theta }{1+\gamma }  \label{a_RE}
\end{equation}

Note that $y$ has no direct payoff relevance for the DM. Moreover, since $y$
does not affect the transmission from $\theta $ and $a$ to $x$, it is
entirely irrelevant for the DM's decision, because $y$ is a consequence of $%
x $ and $a$. Therefore, the constant $\lambda $ plays no role in the DM's
rational action. Moreover, the variances of the noise terms $\varepsilon
,\eta $ are irrelevant for $a_{G^{\ast }}$.\medskip

Now relax the assumption that the DM correctly perceives the true model.
Instead, he believes that the distribution over $\theta ,a,x,y$ obeys a
causal structure that is given by the rightward DAG in Figure 1, which I
denote by $G$. Compared with $G^{\ast }$, the DAG $G$ inverts the causal
link between $x$ and $y$, and also removes the link $a\rightarrow y$. This
means that the DM regards $y$ as an exogenous direct cause of $x$, instead
of the consequence of $x$ and $a$ that it actually is.

The following are two examples of situations that fit this specification of
the true and subjective DAGs $G^{\ast }$ and $G$.\medskip

\noindent \textit{Example 2.1: Parenting}

\noindent Here is a variant on the story presented in the Introduction. A
parent observes characteristics of his child (possibly the child's home
behavior in a previous period), captured by the state variable $\theta $. He
then chooses how toughly to behave toward the child, as captured by the
variable $a$. The child's resulting home behavior is captured by some index $%
x$. The quality of the child's school interactions with teachers and peers,
as measured by the index $y$, is a consequence of the child's and the
parent's home behavior. However, the parent believes that school
interactions are an independent driver of the child's home behavior, rather
than a consequence of home interactions.\medskip

\noindent \textit{Example 2.2: Quantity setting}

\noindent A large firm observes a demand indicator $\theta $ before setting
its production quantity $a$. The price $x$ is a function of demand and
quantity. The variable $y$ represents a competitive fringe that reacts to
the price and the production quantity. However, the firm believes that the
market agents that constitute the competitive fringe are not price takers,
but rather independent decision makers that affect the price.\medskip

The quadratic loss function $u(x,a)=-(x-a)^{2}$ is not necessarily plausible
in these examples. However, in both cases we can write down a plausible
quadratic utility function (e.g. due to linear demand and constant marginal
costs in Example 2.2), and then redefine $a,x,y$ via some linear
transformation to obtain the quadratic loss specification, without any
effect on our formal results (that is, the DM's action will be the same up
to the above linear transformation). The reason is that under any quadratic
utility function, the DM's optimal action is linear in $E_{G}(x\mid \theta
,a)$, which is his subjective expectation of $x$ conditional on $\theta $
and $a$.\medskip

Following Spiegler (2016), I assume that given the long-run distribution $p$%
, the DM forms a subjective belief, denoted $p_{G}$, by fitting his model $G$
to $p$ according to the Bayesian-network factorization formula. For an
arbitrary DAG $R$ over some collection of variables $x_{1},...,x_{n}$, the
formula is%
\begin{equation}
p_{R}(x_{1},...,x_{n})=\dprod\limits_{i=1}^{n}p(x_{i}\mid x_{R(i)})
\label{bayesnetfactor}
\end{equation}%
where $x_{M}=(x_{i})_{i\in M}$ and $R(i)$ represents the set of variables
that are viewed as direct causes of $x_{i}$ according to $R$. Formula (\ref%
{BNfactorization}) was a special case of (\ref{bayesnetfactor}) for $%
R=G^{\ast }$. Given our specification of the subjective DAG $G$, (\ref%
{bayesnetfactor}) reads as follows:%
\[
p_{G}(\theta ,a,x,y)=p(\theta )p(y)p(a\mid \theta )p(x\mid \theta ,a,y) 
\]

The interpretation of this belief-formation model is as follows. The DM's
DAG $G$ is a misspecified subjective model. The DM perceives the statistical
regularities in his environment through the prism of his incorrect
subjective model. In other words, he fits his model to the long-run
statistical data (generated by the true distribution $p$), producing the
subjective distribution $p_{G}$. A related interpretation is based on
Esponda and Pouzo (2016). We can regard $G$ as a set of probability
distributions -- i.e. all distributions $p^{\prime }$ for which $%
p_{G}^{\prime }\equiv p^{\prime }$. The DM goes through a process of
Bayesian learning, by observing a sequence of $i.i.d$ draws from $p$. The
limit belief of this process is $p_{G}$ (see Spiegler (2020a)).

The DM's subjective belief $p_{G}$ induces the following conditional
distribution over $x$, given the DM's information and action:%
\begin{equation}
p_{G}(x\mid \theta ,a)=\int_{y}p(y)p(x\mid \theta ,a,y)
\label{p_G(x|theta,a)}
\end{equation}%
The DM chooses $a$ to maximize his expected utility w.r.t this conditional
belief, which means minimizing 
\begin{equation}
\int_{x}p_{G}(x\mid \theta ,a)(x-a)^{2}\medskip  \label{objective function}
\end{equation}

Spiegler (2016) demonstrates that in general, this behavioral model can be
ill-defined because unlike $p(x\mid \theta ,a)$, the subjective conditional
probability $p_{G}(x\mid \theta ,a)$ need not be invariant to the DM's
strategy $(p(a\mid \theta ))$. To resolve this ambiguity, Spiegler (2016)
introduces a notion of \textquotedblleft personal
equilibrium\textquotedblright , such that the DM's subjective optimization
is formally defined as an equilibrium concept (this is a hallmark of models
of decision making under misspecified models -- see also Esponda and Pouzo
(2016)). However, given our specification of $G$ and $G^{\ast }$, this
subtlety is irrelevant.\footnote{%
Specifically, the reason is that the payoff-relevant variables $\theta ,a,x$
form a clique in $G$, while $G$ treats $y$ as independent of $\theta ,a$.
Spiegler (2016) provides a general condition for specifications of $G$ and $%
G^{\ast }$ that allow the modeler to ignore personal-equilibrium effects.}
Therefore, in what follows I analyze the DM's subjective optimization as
such, without invoking the notion of personal equilibrium. The variation in
Section 4.2 will force us to revisit this issue.$\medskip $

\noindent \textit{Comment: The pure prediction case}

\noindent When $\gamma =0$, the DM's action has no causal effect on $x$. In
other words, the DM faces a pure prediction problem: his subjectively
optimal action is equal to his subjective expectation of $x$ conditional on $%
\theta $. Therefore, the link $a\rightarrow x$ can be omitted from $G^{\ast
} $. This raises the following question: can we be equally cavalier about
whether to include the link $a\rightarrow x$ in the subjective DAG $G$?
Including the link means that the DM erroneously regards his action as a
direct cause of $x$ (but he is open to the possibility that when estimated,
the measured effect will be null).

It turns out that the answer to our question is positive, in the following
sense. When we remove the link $a\rightarrow x$ from $G$, we can no longer
ignore the subtleties that called for treating subjective maximization as an
equilibrium phenomenon. If we apply the concept of personal equilibrium, and
impose in addition the requirement that the DM's equilibrium strategy is
pure, then the model's prediction will be the same as when $G$ includes the
link $a\rightarrow x$. If the relation between $a$ and $\theta $ were noisy
-- e.g. by introducing a random shock to the DM's behavior, or an explicit
preference shock -- this would cease to be true, and we would have to take a
stand on whether $G$ includes the link $a\rightarrow x$ -- i.e. whether the
DM understands that $a$ has no causal effect on $x$.$\medskip $

\noindent \textit{Example 2.3: Inflation forecasting with a wrong Phillips
curve}

\noindent The following variation on Sargent's (1999) simplified version of
the well-known monetary model due to Barro and Gordon (1983) fits the pure
prediction case.\footnote{%
This specification (with rational expectations) was also employed by Athey
et al. (2005).} The variable $\theta $ represents a monetary quantity such
as money growth, which determines inflation $x$ via (\ref{eq x}).
Independently, the private sector (corresponding to the DM in our model)
observes $\theta $ and forms an inflation forecast $a$. The variable $y$
represents employment, which is determined as a function of inflation and
inflationary expectations via the Phillips curve (\ref{eq y}). The parameter 
$\lambda $ measures the extent to which anticipated inflation offsets the
real effect of actual inflation; the case of $\lambda =1$ captures the
\textquotedblleft new classical\textquotedblright\ assumption that only
unanticipated inflation has real effects.

The private sector's inflation forecast is based on a misspecified Phillips
curve, which regards $y$ as an independent, explanatory variable and $x$ as
the dependent variable. Sargent (1999) and Cho and Kasa (2015) refer to this
inversion of the causal inflation-employment relation in terms of an
econometric identification strategy, and dub it a \textquotedblleft
Keynesian fit\textquotedblright . The key difference between these papers
and the present example is that they assume that the private sector has
rational expectations; it is the monetary authority (which chooses $\theta $%
) that operates under a misspecified model (Spiegler (2016) describes how to
reformulate Sargent's example in the DAG language). In contrast, in the
present example, it is the \textit{private sector} that bases its inflation
forecasts on a wrong Phillips curve.

\section{Analysis}

The following is the main result of this paper.$\medskip $

\begin{proposition}
\label{mainresult}The DM's subjectively optimal strategy is%
\begin{equation}
a_{G}(\theta )=\frac{\theta }{1+\gamma +\tau (1-\lambda )}
\label{a(theta,u)}
\end{equation}
\end{proposition}

\begin{proof}
By assumption, the DM chooses $a$ to minimize (\ref{objective function}).
Since $p_{G}$ is Gaussian, the conditional distribution $p_{G}(x\mid \theta
,a,y)$ can be written as a regression equation%
\begin{equation}
x=c_{0}+c_{1}\theta +c_{2}a+c_{3}y+\psi  \label{x_regression}
\end{equation}%
where $c_{0},c_{1},c_{2},c_{3}$ are coefficients and $\psi \sim N(0,\sigma
_{\psi }^{2})$ is an independent noise term. Since $G$ treats $y$ and $\psi $
as independently distributed error terms with $E(\psi )=0$, (\ref%
{x_regression}) implies%
\begin{equation}
E_{G}(x\mid \theta ,a)=c_{0}+c_{1}\theta +c_{2}a+E(y)  \label{EG(x|a,theta)}
\end{equation}%
It also follows from (\ref{x_regression}) that (\ref{objective function})
can be rewritten as%
\begin{equation}
\int_{y}p(y)E_{\psi }(c_{0}+c_{1}\theta +c_{2}a+c_{3}y-a+\psi )^{2}
\label{objective_plug}
\end{equation}%
where the expectation over $\psi $ is taken w.r.t the independent
distribution $N(0,\sigma _{\psi }^{2})$. Note that the DM integrates over $y$
as if it is independent of $a$ because indeed, $G$ posits that $y\perp a$.
This assumption is wrong, according to the true model $G^{\ast }$.

Choosing $a$ to minimize (\ref{objective_plug}), and using the observation
that $E(\psi )=0$, we obtain the following first-order condition:%
\[
(c_{2}-1)\cdot \lbrack c_{0}+c_{1}\theta +c_{2}a+c_{3}E(y)-a]=0 
\]%
This can be written as%
\[
(c_{2}-1)\cdot \lbrack E_{G}(x\mid \theta ,a)-a]=0 
\]%
Therefore, the DM's strategy in a linear personal equilibrium is:%
\begin{equation}
a=E_{G}(x\mid \theta ,a)=\frac{c_{0}+c_{1}\theta +c_{3}E(y)}{1-c_{2}}
\label{strategy}
\end{equation}

I will show that $c_{0}=E(y)=0$, derive $c_{1}$ and $c_{2}$ and show they
are uniquely determined (and in particular that $c_{2}\neq 1$). To do so,
let us derive 
\begin{equation}
E_{G}(x\mid \theta ,a)=\int_{y}p(y)E(x\mid \theta ,a,y)
\label{EG(x|theta,a)}
\end{equation}%
Let us calculate $E(x\mid \theta ,a,y)$. Plugging (\ref{eq x}) in (\ref{eq y}%
), we obtain%
\begin{equation}
x=\theta -\gamma a+\varepsilon =y+\lambda a-\eta  \label{x and y}
\end{equation}%
Therefore,%
\begin{equation}
E(x\mid \theta ,a,y)=E(\theta -\gamma a+\varepsilon \mid \theta ,a,y)=\theta
-\gamma a+E(\varepsilon \mid \theta ,a,y)  \label{E(x|theta,a,y)}
\end{equation}

To derive $E(\varepsilon \mid \theta ,a,y)$, rearrange (\ref{x and y}) to
obtain%
\[
\varepsilon +\eta =y+(\lambda +\gamma )a-\theta 
\]%
Therefore,%
\[
E(\varepsilon \mid \theta ,a,y)=E(\varepsilon \mid \varepsilon +\eta
=y+(\lambda +\gamma )a-\theta ) 
\]%
Since $\varepsilon \sim N(0,\sigma _{\varepsilon }^{2})$ and $\eta \sim
N(0,\sigma _{\eta }^{2})$ are independent variables, we can use the standard
signal extraction formula to obtain%
\begin{equation}
E(\varepsilon \mid \theta ,a,y)=\beta \cdot (y+(\lambda +\gamma )a-\theta )
\label{E epsilon}
\end{equation}%
where%
\begin{equation}
\beta =\frac{\sigma _{\varepsilon }^{2}}{\sigma _{\varepsilon }^{2}+\sigma
_{\eta }^{2}}=\frac{\tau }{1+\tau }  \label{beta}
\end{equation}

Plugging (\ref{E epsilon}) in (\ref{E(x|theta,a,y)}), we obtain%
\[
E(x\mid \theta ,a,y)=(1-\beta )\theta +(\beta \lambda +\beta \gamma -\gamma
)a+\beta y
\]%
Plugging this expression in (\ref{EG(x|theta,a)}), we obtain%
\[
E_{G}(x\mid \theta ,a)=(1-\beta )\theta +(\beta \lambda +\beta \gamma
-\gamma )a+\beta E(y)
\]

Plugging this expression in (\ref{strategy}), we obtain%
\[
a_{G}(\theta )=\frac{1-\beta }{1-\beta \lambda +\gamma (1-\beta )}\theta +%
\frac{\beta }{1-\beta \lambda +\gamma (1-\beta )}E(y)
\]%
Plugging this in (\ref{eq x}) in (\ref{eq y}) and repeatedly taking
expectations, we obtain $E(y)=0$ such that $a_{G}(\theta )$ is given by (\ref%
{a(theta,u)}). We have thus pinned down $c_{1}=0$, $c_{1}=1-\beta $ and $%
c_{2}=\beta \lambda -\gamma (1-\beta )$ in (\ref{strategy}). (The value of $%
c_{3}$ is irrelevant because $E(y)=0$.) Since $\lambda ,\gamma \in \lbrack
0,1]$ and $\beta \in (0,1)$, $c_{2}\neq 1$ such that (\ref{strategy}) is
well-defined and unique.$\medskip $
\end{proof}

The DM's subjectively optimal strategy $a_{G}(\theta )$ has several
noteworthy features, in comparison with the correct-specification benchmark $%
a_{G^{\ast }}(\theta )$.$\medskip $

\noindent \textit{Sensitivity to the state of Nature }$\theta $

\noindent The coefficient of $\theta $ in the expression for $a_{G}(\theta )$
is lower than $1/(1+\gamma )$, which is the correct-specification benchmark
value. This means that the DM's reverse-causality error results in
diminished responsiveness to the state of Nature. The reason is that the DM
erroneously believes that $\theta $ is not the only independent factor
affecting $x$ (apart from $a$). The extent of this rigidity effect depends
on $\tau $, as explained below.$\medskip $

\noindent \textit{The relevance of }$\tau $

\noindent In the correct-specification benchmark, the variances of the error
terms in the relations (\ref{eq x})-(\ref{eq y}) are irrelevant for the DM's
behavior. This is a consequence of the quadratic-normal specification of the
model. In contrast, the DM's reverse-causality error lends these variances a
crucial role. When $\tau =\sigma _{\varepsilon }^{2}/\sigma _{\eta }^{2}$ is
small -- corresponding to the case that the equation for $x$ is precise
relative to the equation for $y$ -- $a_{G}(\theta )$ is closer to the
correct-specification benchmark $a_{G^{\ast }}(\theta )$. In contrast, when $%
\tau $ is large, the departure of $a_{G}(\theta )$ from $a_{G^{\ast
}}(\theta )$ is big. In particular, in the $\tau \rightarrow \infty $ limit, 
$a_{G}(\theta )$ is entirely unresponsive to $\theta $.

The intuition for this effect is as follows. When forming the belief $%
E_{G}(x\mid \theta ,a)$, the DM integrates over $y$ as if it is an
exogenous, independent variable (this is the DM's basic causal error) and
estimates the conditional expectation $E(x\mid \theta ,a,y)$ for each value
of $y$. This conditional expectation represents a signal extraction problem,
and therefore takes into account the relative precision of the equations for 
$x$ and $y$. When $\tau $ is small, $y$ is uninformative of $x$ relative to $%
\theta $, and therefore the conditional expectation $E(x\mid \theta ,a,y)$
places a low weight on $y$. As a result, the DM's erroneous treatment of $y$
does not matter much, such that the DM's belief ends up being approximately
correct. In contrast, when $\tau $ is large, $\theta $ is uninformative of $%
x $ relative to $y$, and therefore the conditional expectation $E(x\mid
\theta ,a,y)$ places a large weight to $y$. Since the DM regards $y$ as an
exogenous variable that is independent of $\theta $, the DM ends up
attaching low weight to the realized value of $\theta $ when evaluating $%
E_{G}(x\mid \theta ,a)$.$\medskip $

\noindent \textit{The role of }$\lambda $

\noindent Recall that the parameter $\lambda $ measures the direct effect of 
$a$ on $y$, which is the outcome variable that the DM erroneously regards as
an independent cause of $x$. The case of $\lambda =0$ corresponds to a DM
who is a \textquotedblleft spectator\textquotedblright\ as far as the
outcome variable $y$ is concerned. At the other extreme, the case of $%
\lambda =1$ corresponds to a DM whose action fully offsets the effect of $x$
on $y$.

It is easy to see from (\ref{a(theta,u)}) that the departure of the DM's
subjectively optimal action from the correct-specification benchmark becomes
smaller as $\lambda $ grows larger. Accordingly, the DM's welfare loss due
to his reverse-causality error becomes smaller. When $\lambda =1$, the DM's
action $coincides$ with the correct-specification benchmark. In this case,
the DM's reverse-causality error inflicts no welfare loss. In other words,
the fact that the DM is not a spectator but an actor who influences $y$
protects him from the consequences of his erroneous view of $y$ as a cause
of $x$.

The intuition for this effect is as follows. When $\lambda =1$, $y=x-a+\eta $%
. This means that given $\theta $,%
\[
y=x-E_{G}(x\mid a,\theta )+\eta 
\]%
In other words, $y$ is equal to the difference between the realization of $x$
and its point forecast, plus independent noise. The point forecast can be
viewed as an OLS estimate based on some linear regression (where $x$ is
regressed on $\theta $, $a$ and $y$, and then $y$ is integrated out). From
this point of view, $y$ is an OLS residual (plus independent noise). By
definition, this residual is independent of the regression's variables, and
therefore its incorrect inclusion has no distorting effect on the estimation
of $x$.

Consider the pure-prediction case of $\gamma =0$, and let us revisit Example
2.3, which concerns inflation forecasts based on a wrong Phillips curve that
inverts the direction of causality between inflation and employment. Recall
that the case of $\lambda =1$ corresponds to the \textquotedblleft new
classical\textquotedblright\ monetary theory that anticipated inflation has
no real effects. As it turns out, in this case the private sector forms
correct inflation forecasts despite its wrong model. A discrepancy with the
rational-expectations forecast -- in the direction of more rigid forecasts
that only partially react to fluctuations in $\theta $ -- arises when $%
\lambda <1$ -- i.e. when anticipated inflation has real effects.

\section{Two Variations}

This section explores two variations on the basic model of Section 2. Each
variation relaxes one of the twin assumptions that $G$ inverts the link
between $x$ and $y$ and treats $y$ as an exogenous variable.

\subsection{Belief in Exogeneity without Reverse Causality}

Let the DM's subjective DAG be the rightward DAG in Figure 1, and modify the
true model $G^{\ast }$ (given by the leftward DAG in Figure 1) with the
following DAG $G^{\ast \ast }$:%
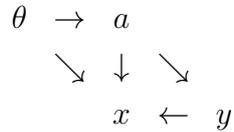
\[
\begin{array}{ccccc}
\theta  & \rightarrow  & a &  &  \\ 
& \searrow  & \downarrow  & \searrow  &  \\ 
&  & x & \leftarrow  & y%
\end{array}%
\]%
Specifically, the equations for $x$ and $y$ are%
\begin{eqnarray*}
y &=&\delta a+\eta  \\
x &=&\theta -\kappa a+\alpha y+\varepsilon 
\end{eqnarray*}%
where $\alpha ,\delta ,\kappa $ are constants in $(0,1)$, and $\varepsilon
\sim N(0,\sigma _{\varepsilon }^{2})$ and $\eta \sim N(0,\sigma _{\eta }^{2})
$ are independently distributed noise variables.

The DM's subjective model $G$ departs from $G^{\ast \ast }$ by removing the
link $a\rightarrow y$, but otherwise the DAGs are identical. Thus, the only
causal misperception that $G$ embodies relative to the true model $G^{\ast
\ast }$ is a belief that $y$ is an exogenous, independent variable. The
following example provides an illustration of this specification of $G$ and $%
G^{\ast \ast }$.\medskip 

\noindent \textit{Example 4.1: Public health}

\noindent Suppose that $\theta $ represents initial conditions of a public
health situation; the DM is a public-health authority and $a$ represents the
intensity of some mitigating public-health measure; $y$ represents the
population's behavioral personal-safety response (higher $y$ corresponds to
more lax behavior); and $x$ represents the eventual public health outcome.
The authority's error is that it takes the likelihood of various scenarios
of the population's behavior as given without taking into account the fact
that it responds to the intensity of the public-health measure.\medskip 

We will now see that although this error has non-null behavioral
implications, those are quite different from what we saw in Section 3. As in
the main model, the DM chooses $a$ such that%
\[
a=E_{G}(x\mid \theta ,a)=\int_{y}p(y)E(x\mid \theta ,a,y)
\]%
The difference is that now, $E(x\mid \theta ,a,y)$ is more straightforward
to calculate because $x$ is conditioned on its actual direct causes:%
\[
E(x\mid \theta ,a,y)=\theta -\kappa a+\alpha y
\]%
I take it for granted that $E(y)=0$. Therefore,%
\[
a=\theta -\kappa a
\]%
such that%
\[
a_{G}(\theta )=\frac{\theta }{1+\kappa }
\]%
By comparison, the correct-specification action would satisfy%
\[
a=E(x\mid \theta ,a)=\theta -\kappa a+\alpha \delta a
\]%
such that%
\[
a_{G^{\ast \ast }}(\theta )=\frac{\theta }{1+\kappa -\alpha \delta }
\]

We see that as in the model of Section 2, treating $y$ as an independent
exogenous variable leads to a more rigid response to $\theta $. Here the
reason is that since $E(y)=0$, the average unconditional effect of $y$ on $x$
is null, whereas conditional on $a>0$ it is positive. Therefore, failing to
treat $y$ as an intermediate consequence of $a$ leads the DM to neglect a
causal channel that affects $x$, which means that he ends up underestimating
the total effect of $a$ on $x$.

However, unlike the basic model of Section 2, the variances of the noise
parameters play no role in the solution: a change in the precision of the
equations for $x$ or $y$ will have no effect on the outcome. Finally, unlike
the basic model, the decision context does not protect the DM from his own
error. Indeed, a larger $\delta $ only $magnifies$ the discrepancy between $%
a_{G}(\theta )$ and $a_{G^{\ast \ast }}(\theta )$. The conclusion is that
the mere failure to recognize the endogeneity of $y$ leads to very different
effects than when this failure is combined with the reverse-causality error.

\subsection{Reverse Causality without Belief in Exogeneity}

In this sub-section I consider a variation on the model that retains the
reverse causality aspect of $G$, while relaxing the assumption that the DM
regards $y$ as an independent exogenous variable. Thus, assume that the true
DAG is $G^{\ast }$, as given by Figure 1, while the DM's subjective DAG $G$
is%
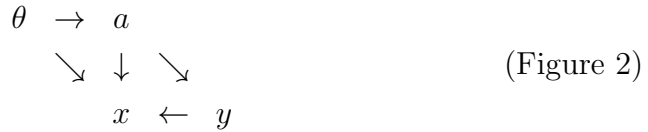
\begin{equation}
\begin{array}{ccccc}
\theta & \rightarrow & a &  &  \\ 
& \searrow & \downarrow & \searrow &  \\ 
&  & x & \leftarrow & y%
\end{array}
\tag{Figure 2}
\end{equation}%
The subjective model departs from the true model by inverting the link
between $x$ and $y$, but otherwise it is identical to $G^{\ast }$.

The DM's misperception can be viewed in terms of the
conditional-independence properties that it gets wrong. While $G$ postulates
that $y\perp \theta \mid a$, $G^{\ast }$ violates this property because of
the causal channel from $\theta $ to $y$ that passes through $x$.\medskip

\noindent \textit{Example 4.2: More parenting}

\noindent Suppose that $\theta $ represents an adolescent's initial
conditions. The DM is the adolescent's parent. The variable $y$ represents
the amount of time that the adolescent spends online, and $a$ represents the
intensity of the parent's attempts to limit this online exposure. The
variable $x$ represents the adolescent's mental health. Thus, according to
the true model, online exposure is not a cause but a consequence of mental
health. The parent's intervention may have a direct effect on mental health,
possibly because it includes spending more time with his child, which may
have direct effects on $x$, independently of the activity it substitutes
away from. The parent believes that online exposure has a direct causal
effect on the adolescent's mental health.\medskip 

Let us analyze the DM's subjectively optimal strategy under this
specification of $G$. The DM chooses $a$ to minimize%
\begin{equation}
\int_{x}p_{G}(x\mid \theta ,a)(x-a)^{2}=\int_{y}p(y\mid a)\int_{x}p(x\mid
\theta ,a,y)(x-a)^{2}  \label{extended pG}
\end{equation}%
Note that the conditional probability term $p(y\mid a)$ is not invariant to
the DM's strategy $(p(a\mid \theta ))$. Therefore, unlike the previous
specifications of the true and subjective DAGs, here we cannot define the
DM's subjective optimization unambiguously. Let us apply the notion of
personal equilibrium as in Spiegler (2016). A DM's strategy $(p(a\mid \theta
))$ is a personal equilibrium if it always prescribes subjective optimal
actions with respect to the belief $(p_{G}(x\mid \theta ,a))$, which in turn
is calculated when we take the DM's strategy as given. I now analyze
personal equilibria for this specification.\footnote{%
In general, this definition requires trembles in order to be fully precise.
This subtlelty is irrelevant for our present purposes.}

Using essentially the same argument as in the proof of Proposition \ref%
{mainresult}, the DM's strategy satisfies:%
\begin{equation}
a=E_{G}(x\mid \theta ,a)  \label{a=EG}
\end{equation}%
where%
\[
E_{G}(x\mid \theta ,a)=\int_{y}p(y\mid a)E(x\mid \theta ,a,y)
\]

Calculating $E(x\mid \theta ,a,y)$ proceeds exactly as in the proof of
Proposition \ref{mainresult} because the true process that links these four
variables is the same. Therefore,%
\[
E(x\mid \theta ,a,y)=(1-\beta )\theta +(\beta \lambda +\beta \gamma -\gamma
)a+\beta y 
\]%
where $\beta $ is given by (\ref{beta}). It follows that%
\begin{eqnarray*}
E_{G}(x &\mid &\theta ,a)=\int_{y}p(y\mid a)[(1-\beta )\theta +(\beta
\lambda +\beta \gamma -\gamma )a+\beta y] \\
&=&(1-\beta )\theta +(\beta \lambda +\beta \gamma -\gamma )a+\beta
\int_{y}p(y\mid a)y
\end{eqnarray*}

Our task now is to calculate the last term of this expression: First,
plugging (\ref{eq x})-(\ref{eq y}), we obtain%
\begin{eqnarray*}
\beta \int_{y}p(y &\mid &a)y=\beta \int_{\theta ^{\prime }}p(\theta ^{\prime
}\mid a)\int_{y}p(y\mid \theta ^{\prime },a)y \\
&=&\beta \int_{\theta ^{\prime }}p(\theta ^{\prime }\mid a)[\theta ^{\prime
}-\gamma a-\lambda a] \\
&=&\beta E(\theta ^{\prime }\mid a)-\beta \left( \gamma +\lambda \right) a
\end{eqnarray*}%
where $E(\theta ^{\prime }\mid a)$ is determined by the DM's equilibrium
strategy, which we have yet to derive.

Plugging this in the expression for $E_{G}(x\mid \theta ,a)$, we obtain%
\begin{eqnarray*}
E_{G}(x &\mid &\theta ,a)=(1-\beta )\theta +(\beta \lambda +\beta \gamma
-\gamma -\beta \gamma -\beta \lambda )a+\beta E(\theta ^{\prime }\mid a) \\
&=&(1-\beta )\theta -\gamma a+\beta E(\theta ^{\prime }\mid a)
\end{eqnarray*}%
Plugging in (\ref{a=EG}) and rearranging, we obtain%
\[
\theta =\frac{(1+\gamma )a-\beta E(\theta ^{\prime }\mid a)}{1-\beta } 
\]%
This equation defines $\theta $ as a deterministic function $f$ of $a$, such
that $E(\theta ^{\prime }\mid a)=f(a)$. It follows that%
\[
f(a)=\frac{(1+\gamma )a-\beta f(a)}{1-\beta } 
\]%
or%
\[
f(a)=(1+\gamma )a 
\]%
Inverting the function, we obtain%
\[
f^{-1}(\theta )=\frac{\theta }{1+\gamma } 
\]%
which is the DM's unique personal equilibrium strategy.

Thus, when the DM correctly perceives that $a$ has a direct causal effect on 
$y$, the fact that he inverts the causal link between $x$ and $y$ is
immaterial for his behavior, which coincides with the correct-specification
benchmark. It follows that the DM's belief that $y$ is independent of $%
\theta $ is essential for the anomalous effects of our main model. The same
conclusion would be obtained if we assumed instead that the DM's subjective
DAG does not contain a direct $a\rightarrow y$ link but instead contains the
link $\theta \rightarrow y$ (or an indirect link $\theta \rightarrow \phi
\rightarrow y$, where $\phi $ is a variable that is objectively correlated
with $\theta $ but independent of all other variables conditional on $\theta 
$).

\section{Conclusion}

This paper explored behavioral consequences of a certain reverse-causality
error -- inverting the causal link between two outcome variables and deeming
one of them exogenous -- in quadratic-normal environments. Two novel
qualitative effects emerge from our analysis. First, the DM's behavior is
sensitive to the variances of the noise terms in the equations for the two
outcome variables $x$ and $y$ -- something that would not arise in the
quadratic-normal setting under a correctly specified subjective model.
Second, the anomalous effects vanish when the DM's direct effect on the
final outcome variable $y$ fully offsets the effect that the intermediate
outcome variable $x$ has on $y$. In this sense, the decision context
protects the DM from his reverse-causality error.

As we saw in Section 4, these novel effects crucially rely on the
combination of the two aspects of the DM's model misspecification -- namely,
his inversion of the causal link between the two outcome variables and his
belief in the exogeneity of $y$. When one of these assumptions is relaxed,
the DM's behavior ceases to display these effects and can even revert to the
correct-specification benchmark. This should not be construed as
\textquotedblleft fragility\textquotedblright\ of our analysis, but rather
as further demonstration that in Gaussian environments,
rational-expectations predictions tend to be robust to many causal
misspecifications (see Spiegler (2020b)). The same causal misperceptions
would be less innocuous under a non-Gaussian objective data-generating
process.

And of course, the specification of the true and subjective DAGs in Figure 1
does not exhaust the range of relevant reverse-causality errors. Spiegler
(2016) gives an example in which the true DAG is $a\rightarrow x\leftarrow y$%
, the DM's wrong model is $a\rightarrow x\rightarrow y$, and $y$ (rather
than $x$) is the payoff-relevant variable. That is, in reality $y$ is
exogenous and the DM's reverse-causality error leads him to believe that $a$
causes $y$ -- the exact opposite of the situation in our main model. Eliaz
et al. (2021) quantify the maximal possible distortion of the estimated
correlation between $x$ and $y$ due to this error (as well as more general
versions of it).

Thus, I am not claiming that the results in this paper are universal
features of situations in which the true and subjective DAGs disagree on the
direction of a certain link. That would be analogous to claiming that
subgame perfect equilibrium in an extensive-form game is robust to changes
in the game form. Nevertheless, as the examples in this paper demonstrated,
many real-life situations fit the mold of Figure 1 and its variations, such
that hopefully the analysis in this paper has provided relevant insights
into the consequences of reverse-causality errors for agents who act on
them.

\end{document}